\documentclass{ijuc}
\usepackage[pdftex]{graphics}
\usepackage{amsmath}    
\usepackage{amssymb}    
\usepackage{amsfonts}   
\usepackage{amsthm}
\usepackage[linesnumbered,ruled,vlined]{algorithm2e}
\newtheorem{definition}{Definition}
\newtheorem{lemma}{Lemma}
\newtheorem{theorem}{Theorem}
\newtheorem{question}{Question}
\begin{document}

\title{Efficient Algorithms for Injectivity and Bounded Surjectivity of One-dimensional Nonlinear Cellular Automata}

\author{Chen Wang\inst{1}
\and Junchi Ma\inst{2} 
\and Weilin Chen\inst{3}
\and Defu Lin\inst{4}
\and Chao Wang\inst{5}\email{wangchao@nankai.edu.cn}
}

\institute{$^{1,2,3,4,5}$College of Software, Nankai University}

\def\received{Received 17 December 2004; In final form 1 April 2005}

\maketitle

\begin{abstract}
Nonlinear cellular automata are extensively used in simulations, image processing, cryptography, and so on. The determination of their fundamental properties, injectivity and surjectivity, related to information loss during the evolution, is necessary in various applications. Currently, people still use Amoroso's algorithms for injectivity and surjectivity determinations, but this incurs significant computational costs when applied to complex nonlinear cellular automata. We have optimized Amoroso's surjectivity algorithm, improving its operational efficiency greatly and extended its applicability to various boundaries. Furthermore, we have introduced new theorems and algorithms for determining injectivity, which offer substantial improvements over Amoroso's algorithm in both time and space. With these new algorithms, we are equipped to determine the properties of larger and more complex cellular automata, thereby employing more advanced cellular automata to achieve increasingly complex functionalities. 
\end{abstract}

\keywords{Cellular automata, Surjectivity, Injectivity, Boundary}

\section{Introduction}
Cellular automaton (CA) is a system of finite automata (cells) that interact locally and operate in parallel. This simplifies and simulates many continuous, complex natural models by analyzing discrete automata evolution.

John von Neumann proposed CA in 1951 to simulate cellular self-replication in biological development \cite{1951_Neumann}. Moore proved that mutually erasable configurations imply the existence of the Garden-of-Eden \cite{1962_Moore}, and Myhill related this to CA surjectivity \cite{1963_Myhill}. In the 1970s, Conway's "Game of Life" used two-dimensional CA to simulate life processes \cite{1970_Convey}. Amoroso and Patt developed algorithms to determine CA surjectivity and injectivity \cite{1972_Amoroso, 1975_Amoroso}. Bruckner showed the necessity of rule balance for CA surjectivity and injectivity \cite{1979_Bruckner}. Wolfram's research in the 1980s advanced one-dimensional CA theory \cite{1983_Wolfram,1984_Wolfram,1984__Wolfram,1986_Wolfram,2002_Wolfram}. Culik, Head, and Sutner further explored CA properties in the late 1980s and early 1990s \cite{1987_Culik, 1989_Head, 1991_Sutner}. Kari's 1994 work highlighted the undecidability of two-dimensional CA injectivity and surjectivity, underscoring the importance of one-dimensional CA decisions \cite{1994_Kari}. Key summaries of CA research are found in Wolfram and Kari's works \cite{1985_Wolfram, 2005_Kari}.

Post-1990s, CA theory matured, leading to diverse CA adaptations for natural and societal phenomena. Notably, linear cellular automata (LCA) were introduced in 1983, simplifying complexity with transition matrices \cite{1983_Ito}. Martín del Rey analyzed LCA reversibility periods \cite{2006_Rey, 2007_Rey, 2011_Rey, 2015_Rey}, while Wang's team used DFA to efficiently calculate LCA periodicity \cite{2015_ChaoWang, 2022_ChaoWang}. Freezing cellular automata (FCA), proposed by Goles in 2015, simulate unidirectional systems like the SIR epidemic model \cite{2015_Goles, 2020_Goles, 2020_Goles2, 2021_Goles}. Number-conserving cellular automata (NCCA), proposed by Wolnik in 2017, excel in simulating object movement and enumeration \cite{2017_Wolnik, 2020_Wolnik, 2022_Wolnik, 2020_Wolnik2}.

No matter the characteristics of specialized CA, surjectivity and injectivity remain inevitable problems. When a new CA is introduced, determining its surjectivity and injectivity is crucial. Although theoretical analysis can be conducted for each CA, it can be time-consuming. Efficient and universal decision algorithms for determining surjectivity and injectivity would be highly beneficial. They can not only serve as a tool for determining the properties of general CA but also be used for exploring and verifying the properties of specific types of CA.

The surjectivity and injectivity determination algorithms proposed by S. Amoroso and Y. N. Patt are foundational and crucial methods. These algorithms aim to identify whether a CA system is surjective or injective, meaning whether the state transitions of the system are unique or cover all possible states. Among these, the surjectivity algorithm is praised for its efficient performance, owing to its use of trees as data structures. However, for the injectivity determination, they opted to use tables as data structures, which led to a significant performance decline compared to the surjectivity algorithm.

Detailed analysis of the injectivity algorithm by S. Amoroso and Y. N. Patt shows that, without the aid of hash tables, its time complexity is actually as high as $O(s^{4m})$ \cite{1972_Amoroso}. It is only when hash tables are introduced as auxiliary structures that the complexity can potentially be reduced to $O(s^{2m})$. Given this, there is potential for improving these algorithms, especially in finding a new algorithm that achieves $O(s^{2m})$ complexity without relying on hash tables. Such an algorithm would not only maintain low complexity but also further reduce space and time consumption, thereby enhancing the efficiency and practicality of CA system analysis.

We improve Amaroso's algorithm for surjectivity and enhance the efficiency of the algorithm greatly. This algorithm is nearly the most authoritative in this field, so it seems that Amoroso's surjectivity algorithm has reached the state of art with the help of our improvement. This algorithm can solve the decision of the global (infinite) surjectivity. Moreover, considering the feasibility and construction difficulty of CA models in practical applications, researchers often shift their focus to studying CA with a finite number of cells, known as finite CA. Unlike the theoretical models of infinite CAs, finite CA models limit the number of cells by introducing boundary conditions, such as null (fixed) boundaries \cite{2006_Rey,2011_Akin,2008_Sahoo,2019_Rey}, periodic boundaries \cite{2004_Nobe,2007_Rey,2007_Bingham}, reflective boundaries\cite{2012_Akin,2014_Akin}. This paper presents decision algorithms for surjectivity compatible with any CA under various boundary conditions inspired by Amaroso \cite{1972_Amoroso}. We also extended the surjectivity determination algorithm to the injectivity problem, arriving at new conclusions and presenting a new determination algorithm. The new injectivity algorithm employs entirely new theorems while retaining the efficient tree structure from the surjectivity algorithm. This new algorithm achieves quadratic time complexity without using hash tables, offering better performance in terms of both time and space costs.

Apart from the introduction, this paper contains six sections. In Section \ref{s2}, we introduce the mathematical definitions of CA and the mathematical concepts required in this paper. In Section \ref{s3}, we simplify Amoroso's algorithm for surjectivity and show its modified version compatible with different boundaries. Through the analysis of Amaroso's algorithm in Section \ref{s4}, we summarize a series of concise theorems and an efficient algorithm for global injectivity, analyze its complexity and compare the actual running time of the two algorithms. Section \ref{s5} shows the decision algorithms for the injectivity of CA with different boundaries, proves the equivalence between global injectivity and injectivity with the periodic boundary and extends our results from nonlinearity to linearity. Section \ref{s6} summarizes all the work of this paper and describes future work and open problems.

\section{Preliminaries \label{s2}}
In this section, we give the formal definitions used in the article. We first define the cellular automata (CA) and relevant problems. Then We define different boundaries of one-dimensional (1d) CA. Finally, we give some convenient mathematical expressions that will appear in our results.
\subsection{Cellular automata and relevant problems \label{s21}}

Cellular automaton (CA) is the most basic and important concept in this paper. It is usually defined by a quadruple: $CA = \{d,S,N,f\}$
\begin{itemize}
	\item $d \in \mathbb{Z}_+$ defines the dimension of CA space, which is generally one-dimensional or two-dimensional. A $d$-dimensional space is denoted as $\mathbb{Z}^d$ in this paper.
	\item $S=\{0,1,\ldots, p-1\}$ is a finite set of states representing a cell's state in CA. At any time, the state of each cell is an element in this set. We use $p$ to count the number of elements in this set.
	\item $N=(\vec n_1, \vec n_2, \ldots , \vec n_m)$ represents neighbor vectors, where $\vec n_i \in \mathbb{Z}^d$, and $\vec n_i \neq \vec n_j$ when $i \neq j$ $(i, j = 1, 2,\ldots, m)$. Thus, the neighbors of the cell $\vec n \in \mathbb{Z}^d$ are the $m$ cells $\vec n + \vec n_i, i = 1, 2,\cdots, m$. The neighborhood size of this CA is denoted as $m$ in this paper. We use $L+1+R$ to express a function containing $L$ neighbor on the left and $R$ neighbor on the right. For example, a basic rule can be expressed as $1+1+1$.
	\item $f$: $S^m \rightarrow S$ is the local rule (or simply the rule). The rule maps the current state of a cell and all its neighbors to this cell's next state.
\end{itemize}

Configuration is a snapshot of all automata in CA at a certain time, which can be represented by the function $c: \mathbb{Z} \rightarrow S$. These cells interact locally and transform their states in parallel. Function $\tau$ is used to represent this global transformation. If there are two configurations $c_1$, $c_2$ and $\tau(c_1) = c_2$, then $c_1$ is a predecessor of $c_2$ while $c_2$ is a successor of $c_1$.

The surjectivity of a $CA = \{d,S,N,f\}$ can be expressed as, exist an configuration $c_1$, for any configuration $c_2$, we have $\tau(c_2) \neq c_1$. That means there exists at least one configuration without any predecessors and we call this kind of configurations "Garden-of-Eden". For a CA with a boundary, we say this CA is surjective iff there is no "Garden-of-Eden" no matter the number of cells.

The injectivity of a $CA = \{d, S, N, f\}$ can be expressed similarly: for any two configurations $c_1, c_2$, if $\tau(c_1)=\tau(c_2)$, then there must be $c_1=c_2$. That means all configurations have one and only one predecessor. For a CA with a boundary, we say this CA is injective iff there is not a configuration that has multiple predecessors.

In this paper, we only focus on all one-dimensional CA $(d = 1)$. The followings are some advanced definitions of CA and some functions required in this paper.

\subsection{Boundaries of 1d CA \label{s22}}
The first CA has no boundary. With the development of CA, it is limited to a finite scale. There are many different boundaries of 1d CA. The most important ones studied at present are the null (fixed) boundary, the periodic boundary, and then the reflective boundary. Now, we will introduce these three boundaries respectively.

Now we suppose the rule of CA is $1+1+1$, Then the leftmost cell has no left neighbor, and the rightmost cell has no right neighbor. In this case, the leftmost and rightmost cells in the next transformation will lose their corresponding value. To avoid this, we need to fill their neighbors outside the boundary with certain values. These three boundaries correspond to three different values of their neighbors.

The null boundary is also called the zero boundary, which always fills ``$0$" in the cell outside the configuration. It is a special kind of fixed boundary. If we use a predetermined sequence such as ``$01010...01010$" in these cells, it will be a fixed boundary. The example of a null (fixed) boundary is shown in Figure \ref{Fig2-1}.
\begin{figure}[h]
	\begin{center}
		\scalebox{0.5}{\includegraphics{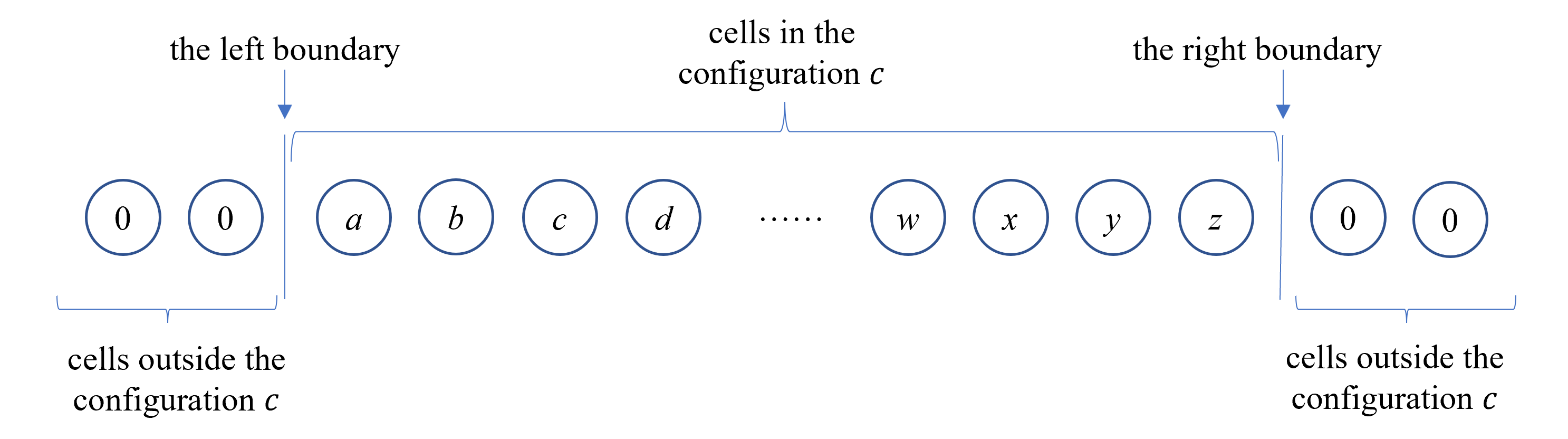}}
	\end{center}
	\caption{the null boundary of 1d CA}
	\label{Fig2-1}
\end{figure}

Periodic boundary regards the configuration as a circle, which fills the cell left-outside the configuration according to the cell in the left of the configuration. The example of the periodic boundary is also shown in Figure \ref{Fig2-2}.

\begin{figure}[h]
	\begin{center}
		\scalebox{0.5}{\includegraphics{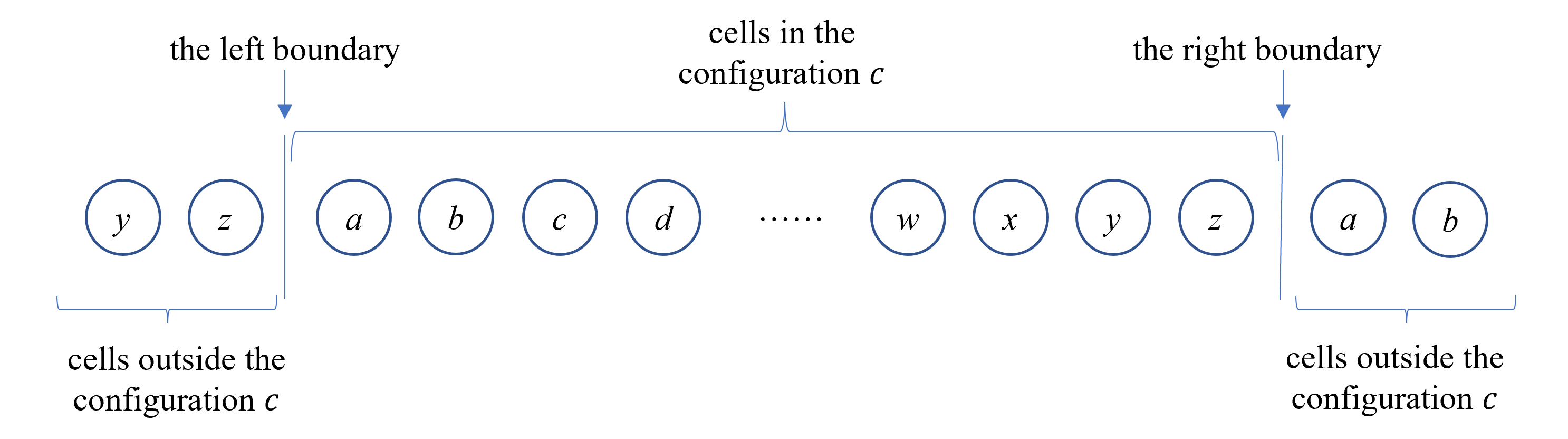}}
	\end{center}
	\caption{the periodic boundary of 1d CA}
	\label{Fig2-2}
\end{figure}

Reflective boundary regards the boundary as a mirror and fills the cell left-outside the configuration according to the cell in the reflective position of the mirror (boundary). We also show the example of a reflective boundary in Figure \ref{Fig2-3}.
\begin{figure}[h]
	\begin{center}
		\scalebox{0.5}{\includegraphics{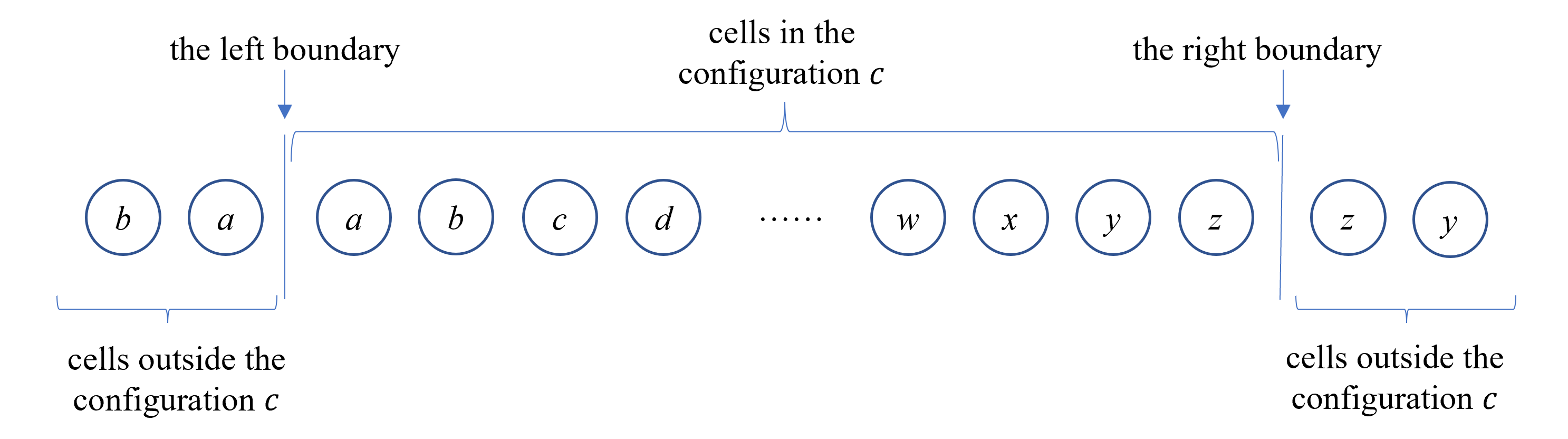}}
	\end{center}
	\caption{the reflective boundary of 1d CA}
	\label{Fig2-3}
\end{figure}

It is noteworthy that if a 1d CA is expressed as ``$L+1+R$", the number of cells we need to replenish is: $L$ on the left and $R$ on the right.

\subsection{Some mathematical expressions and functions \label{s23}}
In this subsection, we will introduce some useful expressions and functions for this article. These definitions are not very important in CA, but they can greatly simplify the expression of the article and make it easier to understand.

Local configuration: We can get a continuous and finite part from an infinite configuration $c$. This finite part can be regarded as a string, called local configuration.

Left (right) side: Because of the limited length, there is a left side and a right side in each local configuration. Supposing a local configuration $c_f=a_1a_2 \cdots a_m$, its left $n$ bits $(n \leq m)$ are $a_1a_2 \cdots a_n$. The same goes for the right. We can use $left_n(c_f)$ to get left $n$ bits of $c_f$, and $right_n(c_f)$ to get right $n$ bits of $c_f$. Figure \ref{Fig2-4} shows an example of this function.
\begin{figure}[h]
	\begin{center}
		\scalebox{0.5}{\includegraphics{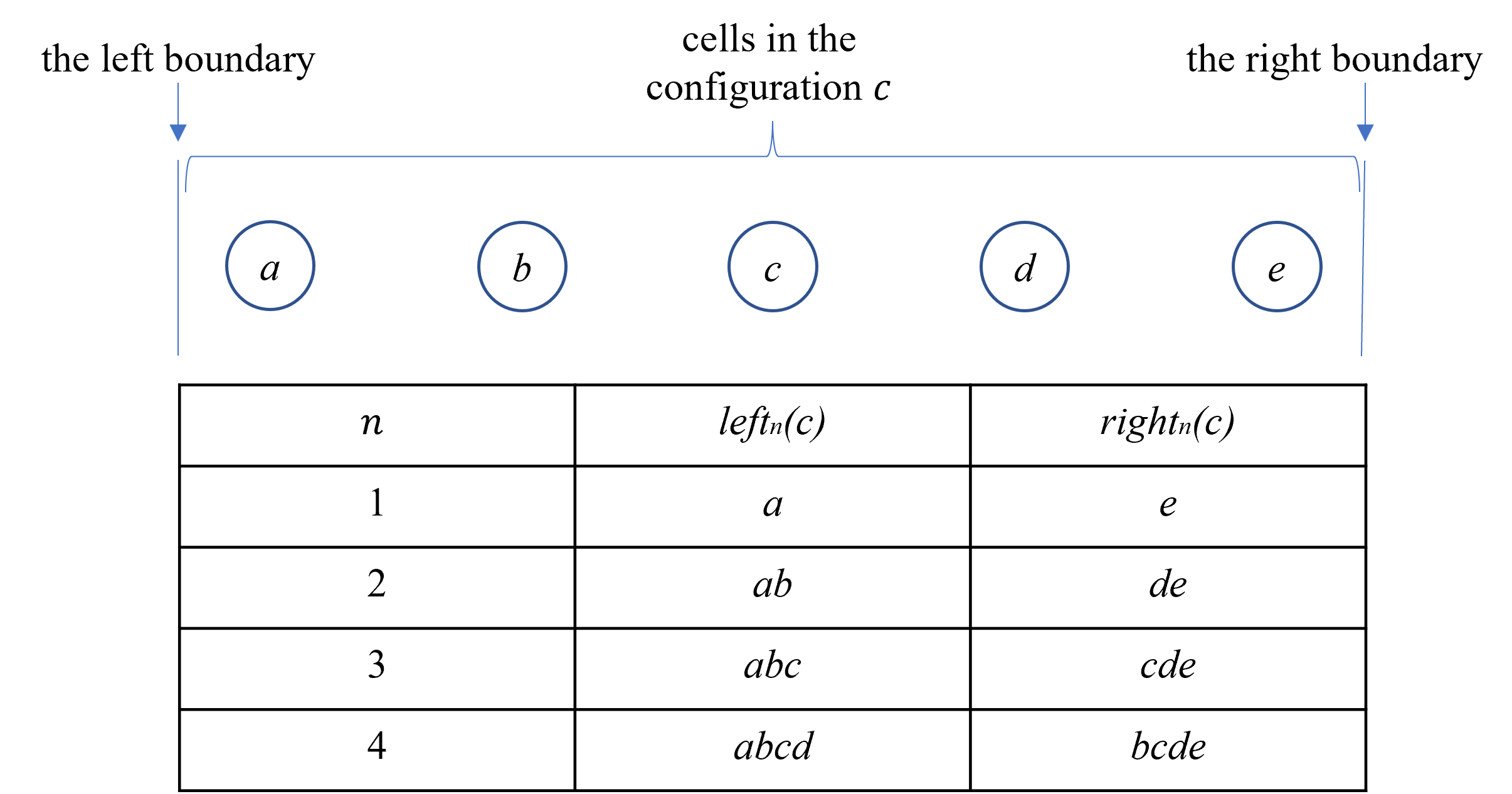}}
	\end{center}
	\caption{an example of function $left_n(c)$ and $right_n(c)$}
	\label{Fig2-4}
\end{figure}

Sequent set: Amoroso has provided the definition of a sequent set to solve the problem of injectivity \cite{1972_Amoroso}. Here we repeat it briefly. On the premise that the neighborhood size is $m$, the right sequent set of local configurations $\alpha$ and $\beta$ is a two-element set $\{\alpha', \beta'\}$ which meets the following conditions:
\begin{itemize}
	\item $right_{m-1}(\alpha) = left_{m-1}(\alpha')$
	\item $right_{m-1}(\beta) = left_{m-1}(\beta')$
	\item $f(\alpha') = f(\beta')$
\end{itemize}
Similarly, the left sequent set of local configurations $\alpha$ and $\beta$ is a two-element set $\{\alpha', \beta'\}$ which meets the following conditions:
\begin{itemize}
	\item $left_{m-1}(\alpha) = right_{m-1}(\alpha')$
	\item $left_{m-1}(\beta) = right_{m-1}(\beta')$
	\item $f(\alpha') = f(\beta')$
\end{itemize}

\section{Surjectivity with different boundaries \label{s3}}
In this section, we first introduce Amoroso's algorithm for surjectivity and its improvement. Then surjectivity decision algorithms with null (fixed) boundary and reflective boundary are given respectively. The surjectivity decision algorithm with the periodic boundary is similar to our new algorithm for global injectivity, so we introduce it in the next section together. Although the examples in this paper is over $\mathbb{F}_2$, all algorithm in this paper is compatible over $\mathbb{F}_P$

\subsection{Amoroso's algorithm for surjectivity and its improvement \label{s31}}
Amoroso's algorithm for surjectivity can be described in brief step by step. And Figure \ref{Fig3-1} shows an example of Amoroso's decision algorithm for surjectivity.
\begin{description}
	\item[Step 1] For a CA with neighborhood size $m$ and state set $S=\{0,1,\ldots, p-1\}$, get the local map table of the rule.
	\item[Step 2] If $f(m)=0$, add $m$ into the root node.
	\item[Step 3] For each node $M$ in layer $i$ $(i \geq 0)$, construct its children $M_0, M_1,\cdots, M_{p-1}$. For each tuple $a_1a_2 \cdots a_{m}$ in $M$ and all $0 \leq j<p$, if $f(a_2a_3 \cdots a_mj)=t$, then $a_2a_3 \cdots a_mj$ is added to $M_t$. If a node is the same as a node that has appeared before, its children will not be constructed. If there are not any m-tuples in a node, then the CA has a Garden-of-Eden which decides the CA is not surjective and the algorithm returns $False$.
	\item[Step 4] If the three steps above are completed and the CA is not decided as non-surjective, then the CA is surjective.
\end{description}

\begin{lemma}
	\label{ts}
	\textbf{The length of tuples in nodes can be reduced from $m$ to $m-1$.}
\end{lemma}
\begin{proof}
	To simplify this algorithm, we need to guarantee the algorithm that the algorithm can still run correctly. So we need to prove three things.
	\begin{itemize}
		\item First, the configuration c and its original image can still correspond correctly.
		\item Second, the place where the original algorithm has empty nodes is still empty after simplification.
		\item Third, the non-empty nodes of the original algorithm are still non-empty after improvement.
	\end{itemize}
	
	It is very easy to prove the first point. Even if the tuple length is reduced from $m$ to $m-1$, we can still use the extra bit of tuples in its children to obtain all parameters required by $f$. The application of the de Bruijn graph in CA uses the same principle \cite{1991_Sutner}.
	
	Secondly, we only retain the last $m-1$ bits of each tuple of the root node and remove the same tuples generated thereby. The number of tuples of the last $m-1$ bits contained in the simplified root node is the same as the original node. If all nodes are constructed in the same way, the number of tuples of the last $m-1$ bits contained in its child node is still the same as the original node in the corresponding position. If a node is empty, the node at the same location should contain all the last $m-1$ bits of the node, that is, 0, and it will also be empty. Similarly, if a node is not an empty node, the node at the same location is also not empty.
\end{proof}

\begin{figure}[h]
	\begin{center}
		\scalebox{0.35}{\includegraphics{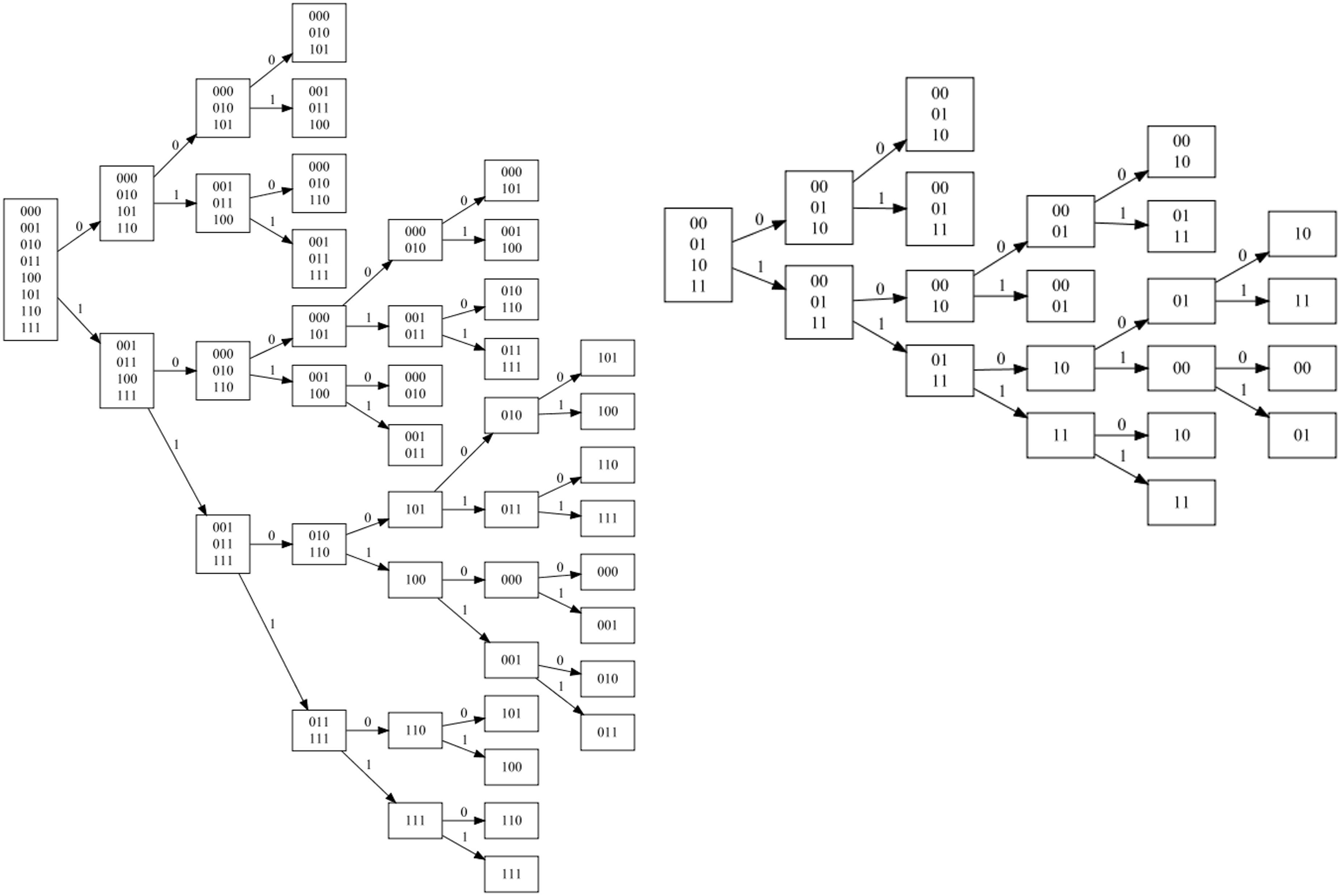}}
	\end{center}
	\caption{the simplification of Amoroso's algorithm for surjectivity (Amaroso's on the left and our simplified version on the right}
	\label{Fig3-1}
\end{figure}
This lemma seems to simply shorten the tuple length from $m$ to $m-1$. However, it decreases the complexity of nodes from $2^{p^m}$ to $2^{p^{m-1}}$, which differs $2^{(p-1)*p^{m-1}}$ times. For rules without Garden-of-Eden, the running time of the algorithm will be much shorter. Figure \ref{Fig3-1} and Algorithm \ref{a1} shows an example of this improvement. We can find that the structure of the tree has been greatly simplified.

\begin{algorithm}[h]
	\small
	\SetAlgoLined
	\label{a1}
	\KwData{local rule $f$}
	\KwResult{whether the CA is surjective}
	let $L$ be an empty queue\;
	let $S_{root} \leftarrow \{a_1...a_{m-1} | a_1,...,a_{m-1} \in \{0,1,...,p-1\}\}$ be the root set\;
	push $S_{root}$ to the back of $L$\;
	\While{$L$ is not empty}{
		remove the headset in $L$ and mark it as $S_{current}$\;
		\If{$S_{current}$ is empty}{
			return the CA is not surjective\;
		}
		\If{$S_{current}$ is a new set}{
			let $S_0$,$S_1$,...,$S_{p-1}$ be empty sets\;
			\For{each $a_1a_2...a_{m-1} \in S_{current}$, each $d \in \{0,1,...,p-1\}$ and each $b \in \{0,1,...,p-1\}$}{
				\If{$f(a_1...a_{m-1}d) = b$} {
					add $a_2...a_{m-1}d$ into $S_b$\;
				}
			}
			
			push $S_0$,$S_1$,...,$S_{p-1}$ to the back of $L$\;
		}
	}
	return the CA is surjective\;
	\caption{decision algorithm for surjectivity}
\end{algorithm}

\subsection{Surjectivity with the fixed boundary \label{s32}}
The overall change between global surjectivity and surjectivity null boundary is similar, the difference is that two boundaries are added (left and right). In the global surjectivity algorithm, we use the leftmost bit in the root node to represent the leftmost cell in the configuration, and the rightmost bit in the leaf nodes to represent the rightmost cell. According to our introduction to the boundary in Subsection \ref{s22}, the process we fill the neighbors of the cell near the boundaries is equal to the limitation to the tuple in the root node and leaf nodes. It is noteworthy that our improvement to the algorithm for surjectivity in the previous subsection produces a marked effect.
\begin{definition}
	Given a CA with neighborhood size $L+1+R$ and the kind of fixed boundary $b_l$ and $b_r$ ($b_l=b_r = 000...000$), the initial set ($s_i$) is the set of $(L+R)$-tuples $t$ which satisfies the following conditions:
	\begin{itemize}
		\item $left_l(t) = b_l$
	\end{itemize}
	the terminal set ($s_t$) is the set of $(L+R)$-tuples $t$ which satisfies the following conditions:
	\begin{itemize}
		\item $left_r(t) = b_r$
	\end{itemize}
\end{definition}
\begin{figure}[h]
	\begin{center}
		\scalebox{0.7}{\includegraphics{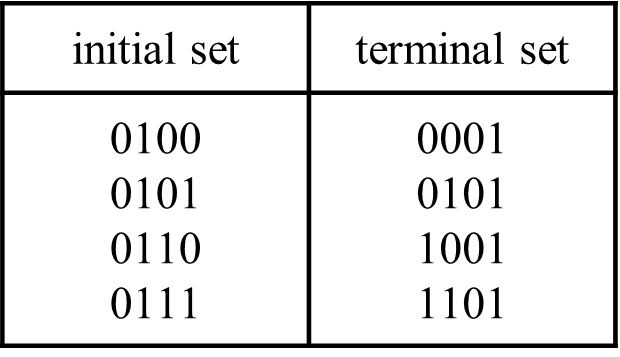}}
	\end{center}
	\caption{the initial set and terminal set for CA with neighborhood size 2+1+2 (fixed boundary 010...10)}
	\label{Fig3-2}
\end{figure}

An example of the initial set and terminal set is shown in Figure \ref{Fig3-2}. We can use this two to determine the surjectivity with the fixed boundary.
In the line 2 of Algorithm \ref{a1}, we have:
\begin{itemize}
	\item let $S_{root} \leftarrow \{a_1...a_{m-1} | a_1,...,a_{m-1} \in \{0,1,...,p-1\}\}$ be the root set
\end{itemize}
Now we modify it into:
\begin{itemize}
	\item let $S_{root} \leftarrow s_i$ be the root set
\end{itemize}
In the line 6 of Algorithm \ref{a1}, we have:
\begin{itemize}
	\item If $S_{current}$ is empty
\end{itemize}
Now we modify it into:
\begin{itemize}
	\item If $S_{current} \setminus s_t = \emptyset$
\end{itemize}

Now, we have a new algorithm to determine the surjectivity with the fixed boundary. Its complexity is equal to the global surjectivity algorithm, and its actual running time is less. Because the reduction of the number of tuples in the root node will reduce the scale of the tree, and the existence of the terminal set will also lead to the earlier termination of the algorithm. Table \ref{table1} lists all surjective CA with neighborhood sizes less than 5 (fixed boundary). In this table, we can find the setting of the boundary affects the surjectivity of the CA little.

\subsection{Surjectivity with the reflective boundary \label{s33}}
It seems like we can use the same way in the reflective boundary. However, for the reflective boundary, we can get the initial set and the terminal set when $L = R$. (The neighborhood size of this CA is $L+1+R$.) Because the initial set and the terminal set are composed of ($L+R$)-tuples. The cell outside the boundary is reflected by those near the boundary. According to Subsection \ref{s22}, we need $L$ extra cells on the left and $R$ cells on the right. If we set the initial set and the terminal set in the same way as the algorithm in the previous subsection, the length of the tuples in the initial set should be $2L$ and the length of the tuples in the terminal set should be $2R$. If $L \neq R$, the algorithm will not work successfully.
An additional algorithm is needed to solve this problem. We can make a preprocess and change the expression of the rule to make $L=R$.
We have introduced the rule $f$ of CA is $S^m \rightarrow S$. Now we suppose:
\begin{itemize}
	\item $m = L+1+R$
	\item $L=R+1$
	\item $f(x_{-L},x_{-L+1},...,x_{-1},x_0,x_1,...x_R) = x_0'$ is the rule.
\end{itemize}
We can add an extra parameter $x_{R+1}$ and keep the following equation regardless of the value of $x_{R+1}$. Now we have a symmetrical rule. We call this the extension of the rule.
\begin{equation}
	\begin{aligned}
		f(x_{-L},x_{-L+1},...,x_{-1},x_0,x_1,...x_R) = & f'(x_{-L},x_{-L+1},...,x_{-1},x_0,x_1,...x_R,x_{R+1})
	\end{aligned}
\end{equation}
It's easy to extend a rule. We need to discuss according to $L$ and $R$ which is bigger. Here we use the Wolfram number to express the local rule $f$ \cite{1984_Wolfram}.
\begin{itemize}
	\item If $L = R+1$ Duplicate each number behind itself (duplicate bit by bit), eg. $00011101 \rightarrow 0000001111110011$
	\item If $R = L+1$ Duplicate the function behind itself (duplicate entirely), eg. $00011101 \rightarrow 0001110100011101$
\end{itemize}
\begin{figure}[h]
	\begin{center}
		\scalebox{0.7}{\includegraphics{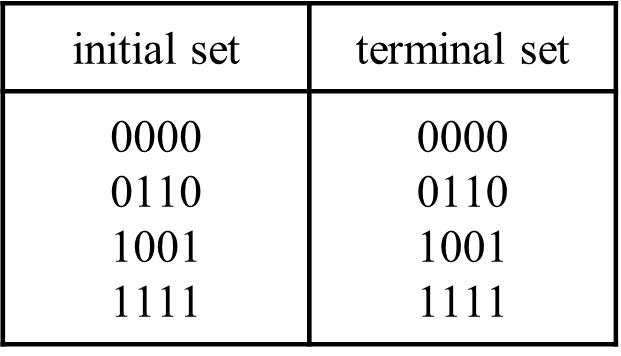}}
	\end{center}
	\caption{the initial set and terminal set for CA with neighborhood size 2+1+2 (reflective boundary)}
	\label{Fig3-3}
\end{figure}
\begin{algorithm}[h]
	\SetAlgoLined
	\label{a2}
	\KwData{the local rule $f$ (a Wolfram number)}
	\KwResult{a symmetrical rule $f'$ (a Wolfram number)}
	Initialize $f' \leftarrow f$;
	
	\While {$L > R$}{
		duplicate $f'$ bit by bit;\\
		$L \leftarrow L-1$;
	}
	\While {$R > L$}{
		duplicate $f'$ entirly;\\
		$R \leftarrow R-1$;
	}
	return $f'$;
	\caption{The extension of a rule}
\end{algorithm}

Now we have completed the operation of extending the rule by one bit. Now we give the algorithm of transforming the rule to $L=R$ in Algorithm \ref{a2}. An example of the initial set and terminal set is shown in Figure \ref{Fig3-3} and the results are shown in table \ref{table1}.

\begin{table}[h]
	\caption{number of surjective CA under different boundary conditions}
	\label{table1}
	\begin{tabular}{ccc}
		\hline
		neighborhood size & boundary set & rule number \\
		\hline
		3 & null boundary & 6\\

		3 & reflective boundary & 2 \\

		3 & periodic boundary & 6\\

		4 & null boundary & 34 \\

		4 & fixed boundary & 34 \\

		4 & reflective boundary & 2\\

		4 & periodic boundary & 16 \\

		\hline
	\end{tabular}
\end{table}

\section{Injectivity tree algorithm \label{s4}}
\subsection{Introduction of Amoroso's algorithm}

Because it is hard to show Amoroso's algorithm for injectivity completely, here we describe it in brief for each step and one of his figures is shown below. This algorithm requires that the rules are balanced.
\begin{description}
	\item[Step 1] Classify all $m$-tuples into different classes according to the local mapping rule $f$. Each class has $p^{m-1}$ tuples.
	\item[Step 2] Construct the sequent table for each class with the first $p^{m-1}-1$ elements of the class as the abscissa axis and the last $p^{m-1}-1$ elements as the ordinate axis.
	\item[Step 3] Search sequent sets for each box in the sequent table. If one of its sequent sets is itself, then the CA is not injective. If the right $m-1$ bits of two $m$-tuples in the box are the same, put $\bigotimes$ into the box. The box is crossed-out if there is neither a sequent set nor $\bigotimes$ in the box.
	\item[Step 4] If a box is crossed-out, it cannot be a sequent set of other boxes. Delete such sequent sets in all boxes and then there will generate new cross-out boxes. Iterate this process as far as possible until there is no change.
	\item[Step 5] Assign values to all non-crossed-out boxes. A symbol $\bigotimes$ is assigned weight 0. If sequent sets of a box are $W_1, W_2, W_3, \cdots, W_n$, the box is assigned weight $1 + max (W_1, W_2, W_3, \cdots, W_n)$. If there are any non-assigned boxes left, the CA is not injective.
	\item[Step 6] After step 5, check all non-crossed-out boxes. If the left $m-1$ bits of the two $m$-tuples in the box are the same, then the CA is non-injective. If the above steps cannot decide that the CA is non-injective, it is injective.
\end{description}

\begin{figure}[h]
	\begin{center}
		\scalebox{0.5}{\includegraphics{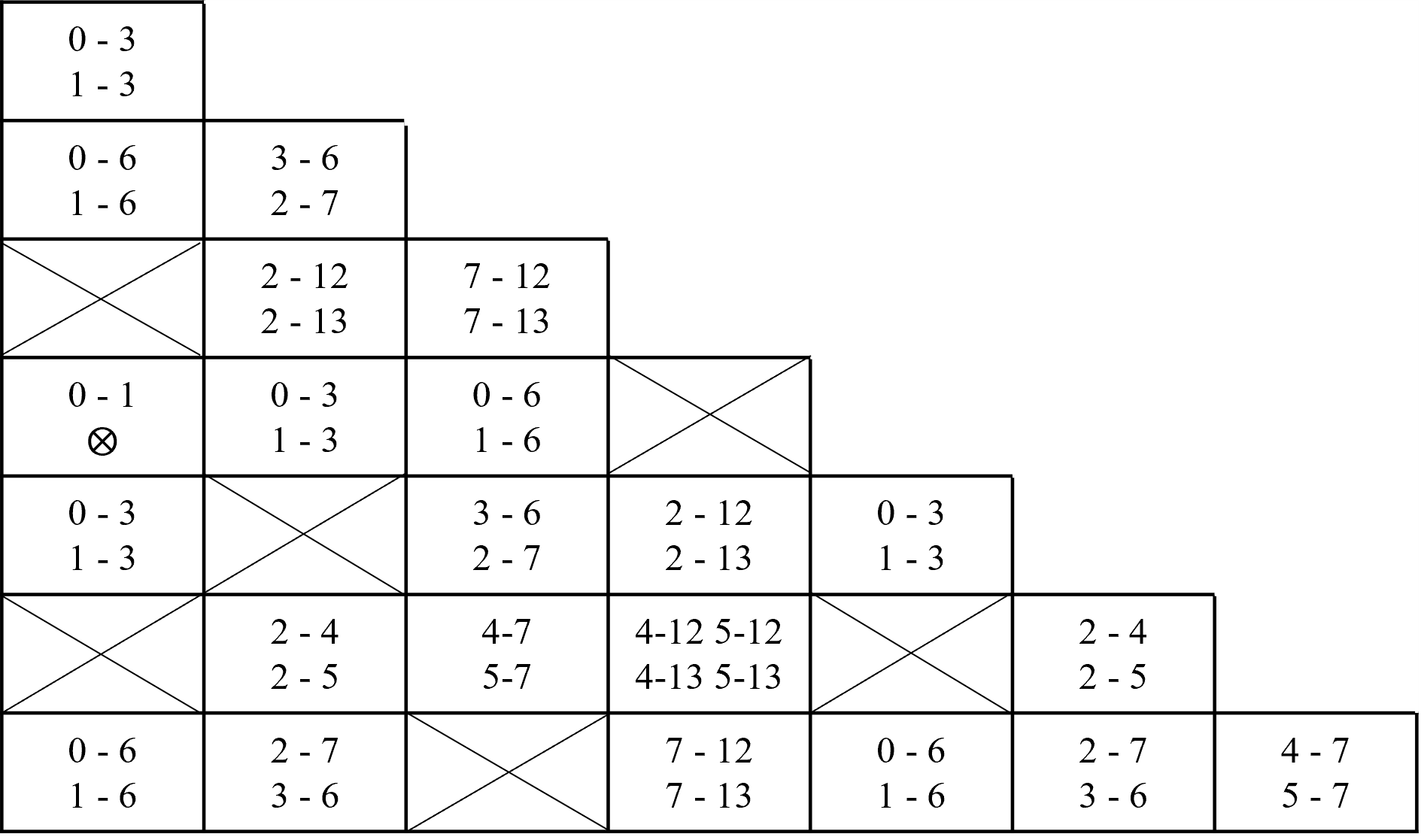}}
	\end{center}
	\caption{the step 3 of Amoroso's algorithm for injectivity}
	\label{a}
\end{figure}

Amoroso has given a theorem along with its proof for the injectivity of CA. The theorem is briefly described as follows. If two configurations have the same successor, their values of automata must be the same after a certain distance which is equivalent to injectivity. Besides, the algorithm has six steps in total and judges whether a CA is injective in three positions. Compared with the theorem, this algorithm is somewhat complex. Although we use hashmap to optimize step 4 greatly, the running complexity is still high.

It is still possible to optimize Amoroso's decision algorithm for injectivity. However, the decision algorithm for surjectivity and the thought of the sequent set given by Amoroso have great reference values. The decision algorithm for surjectivity has excellent efficiency, and the sequent set seems to be the only way to solve the injectivity problem. Therefore, if we can use the tree structure to decide the injectivity, it will not only be easier to understand but also have a remarkable breakthrough in efficiency.

Amoroso's algorithm needs to traverse all sequent sets, which is found unnecessary after our research. After the iterative deletion of crossed-out boxes, the remaining sequent sets will be very few compared with the initial. That is, there may be very few non-injective replaceable sequent sets. When we use the tree structure, the corresponding configurations in the same node have the same successor, so we only need to decide whether the tuples in the nodes can lead to many-for-one. But correspondingly, the surjectivity tree cannot complete this work without restructuring. The tree structure used in the new algorithm for injectivity is different from that in the algorithm for surjectivity.

In addition, Amoroso's theorem for injectivity mentions "a distance". However, this distance is quite long, as long as $(^{2^{m-1}}_2)$, which is impossible for trees to construct. Therefore, to match the tree structure, it is necessary to find and summarize a new theorem for injectivity, which can play a role in the finite but irregular structure of the tree and complete the decision for injectivity.

\subsection{Theorems for injectivity}
\begin{definition}
	If two local configurations $\alpha$ and $\beta$ meet the following conditions, then $\alpha$ and $\beta$ are \textbf{replaceable local configurations}.
	\begin{itemize}
		\item $\alpha$ and $\beta$ contain no less than $2m-1$ cells
		\item $\alpha \neq \beta$
		\item $left_{m-1}(\alpha) = left_{m-1}(\beta)$
		\item $right_{m-1}(\alpha) = right_{m-1}(\beta)$
		\item $\alpha$ and $\beta$ have the same successor according to the local rule $f$
	\end{itemize}
\end{definition}

\begin{lemma}
	\label{l1}
	For a CA with neighborhood size $m$, if two configurations $c_1$ and $c_2$ contain \textbf{replaceable local configurations} $\alpha$ and $\beta$ respectively, then this CA is not injective.
\end{lemma}
\begin{proof}
	After getting the two local configurations $\alpha$ and $\beta$, we can fill ``$0$" at all other locations and get two configurations $c_1$ and $c_2$. It is obvious that $c_1 \neq c_2$. Because $\alpha$ and $\beta$ are replaceable, they have the same successor. As shown in Figure \ref{Fig5-1}, $a_i$ and $b_i$ are the $m-1$ bits of left side and right side of $\alpha$ and $\beta$, $x_i$ and $y_i$ $(i \in \mathbb{Z}_+)$ are mutually replaceable parts. Then it is easy to prove $\tau(c_1)=\tau(c_2)$, so the CA is not injective.
\end{proof}

\begin{figure}[h]
	\begin{center}
		\scalebox{0.07}{\includegraphics{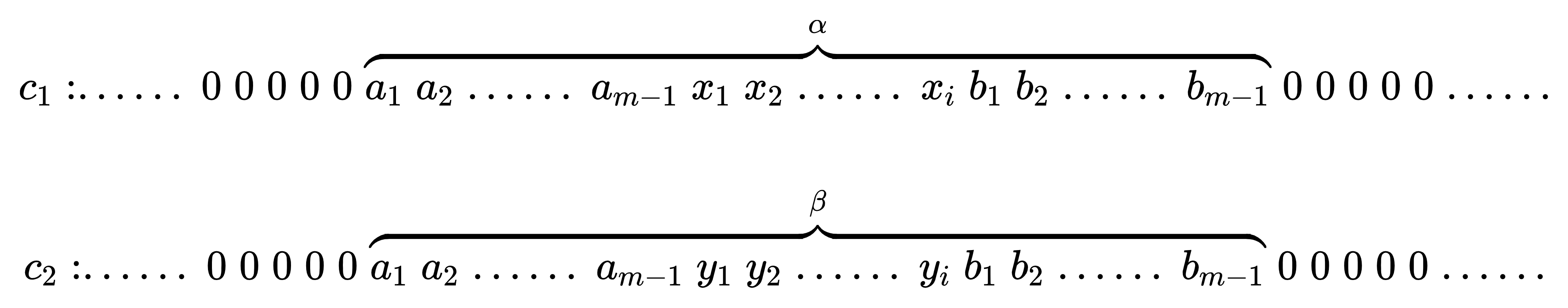}}
	\end{center}
	\caption{replaceable local configurations}
	\label{Fig5-1}
\end{figure}

\begin{definition}
	If two local configurations $\alpha$ and $\beta$ meet the following conditions, then $\alpha$ and $\beta$ are \textbf{periodic local configurations}.
	\begin{itemize}
		\item $\alpha$ and $\beta$ contain no less than $m$ cells
		\item $\alpha \neq \beta$
		\item $left_{m-1}(\alpha) = right_{m-1}(\alpha)$
		\item $left_{m-1}(\beta) = right_{m-1}(\beta)$
		\item $\alpha$ and $\beta$ have the same successor according to the local rule $f$
	\end{itemize}
\end{definition}

\begin{lemma}
	\label{l2}
	For a CA with neighborhood size $m$, if two configurations $c_1$ and $c_2$ contain \textbf{periodic local configurations} $\alpha$ and $\beta$ respectively, then this CA is not injective.
\end{lemma}
\begin{proof}
	We can get local configurations $\alpha'$ and $\beta'$ from $\alpha$ and $\beta$ respectively by removing their left $m-1$ bits. Now, we can construct two configurations $c_1$ and $c_2$, as shown in Figure \ref{Fig5-2}. It is easy to prove $\tau(c_1)= \tau(c_2)$, so the CA is not injective.
\end{proof}

\begin{figure}[h]
	\begin{center}
		\scalebox{0.35}{\includegraphics{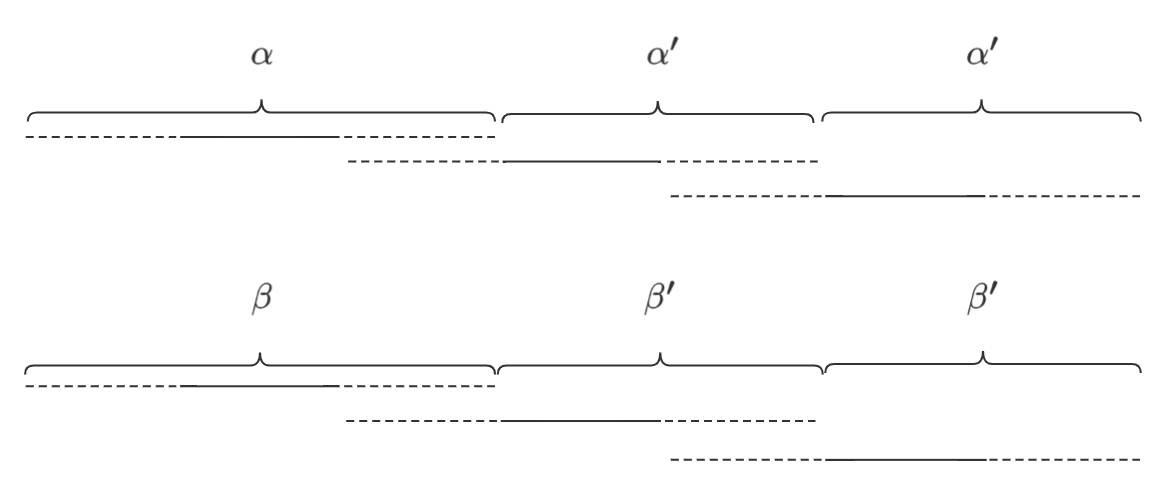}}
		\scalebox{0.05}{\includegraphics{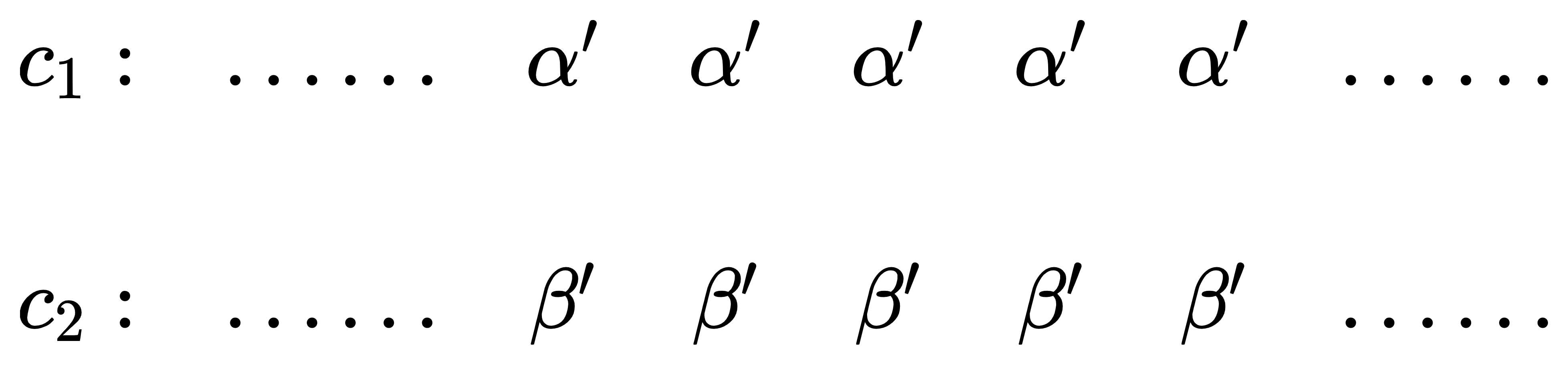}}
	\end{center}
	\caption{periodic local configurations}
	\label{Fig5-2}
\end{figure}

\begin{lemma}
	\label{l3}
	For a CA with neighborhood size $m$, if any two configurations $c_1$ and $c_2$ contain \textbf{neither replaceable nor periodic local configurations}, then this CA is injective.
\end{lemma}
\begin{proof}
	Randomly select two different local configurations $\alpha$ and $\beta$ with the length of $m-1$. Search the left sequent set $T_{-1} = \{\gamma_{-1}, \delta_{-1}\}$ of $\alpha$ and $\beta$ and get new local configurations $\alpha'$ and $\beta'$. Keep looking for a series of left sequent sets $T_{-1}, T_{-2} ,\cdots, T_{-n}$. There are not any periodic local configurations, so we have $T_{-i} \neq T_{-j}$ $(0 < i < j \leq n)$. There are at most $2*(^{2^{m-1}}_2)$ different left sequent sets $T$ in total which are finite. So after looking for at most $2*(^{2^{m-1}}_2)+1$ left sequent sets, there must be $T_{-i}=\{\gamma_{-i}, \delta_{-i}\}$ and $left_{m-1}(\gamma_{-i})=left_{m-1}(\delta_{-i})$. Keep looking for a series of right sequent sets $T_1, T_2 ,\cdots, T_r$. There will not exist $T_i=\{\gamma_i, \delta_i\}$ that $right_{m-1}(\gamma_i)=right_{m-1}(\delta_i)$, otherwise we can find replaceable local configurations. There will also not exist $T_i = T_j$ $(0 < i < j \leq r)$, otherwise we can find periodic local configurations. Now There are at most $2*(^{2^{m-1}}_2)$ types of right sequent set. That is, after searching left sequent sets with the same $m-1$ bits on the left and then searching at most $2*(^{2^{m-1}}_2)$ right sequent sets, we cannot continue to search right sequent sets. No matter which local configurations are selected initially, these two configurations $c_1$ and $c_2$ cannot be constructed. Therefore, the non-existence of replaceable or periodic local configurations is equal to the non-existence of two different configurations with the same successor, which is the definition of injectivity.
\end{proof}

\begin{theorem}
	\label{t1}
	\textbf{The non-existence of replaceable and periodic local configurations is equivalent to global injectivity.}
\end{theorem}
\begin{proof}
	In lemma \ref{l1} and lemma \ref{l2}, We have proved that for a CA with neighborhood size $m$, if replaceable or periodic local configurations exist, then the CA is not injective. In lemma \ref{l3}, it is proved that if replaceable and periodic local configurations do not exist, the CA is injective. Therefore, the non-existence of replaceable and periodic local configurations is equivalent to injectivity.
\end{proof}

\begin{lemma}
	\label{l4}
	If there exist replaceable local configurations, there must be periodic local configurations.
\end{lemma}
\begin{proof}
	According to Figure \ref{Fig5-1}, if there exist replaceable local configurations, we can duplicate the left $m-1$ bit on the right like Figure \ref{Fig5-3}. Local configurations $c_1$ and $c_2$ meet the definition of periodic local configurations. So we only need to search the periodic local configurations in the algorithm to determine global injectivity.
\end{proof}

\begin{figure}[h]
	\begin{center}
		\scalebox{0.07}{\includegraphics{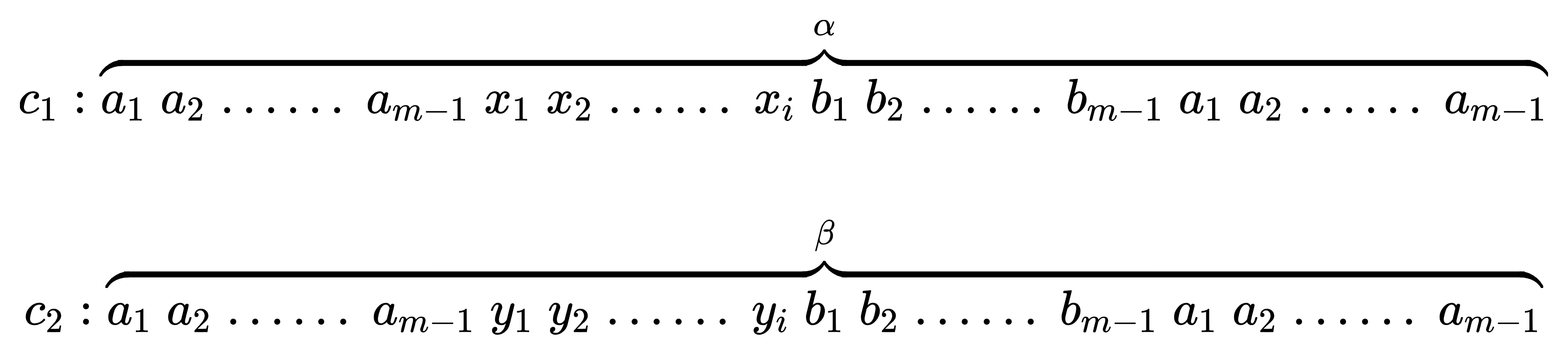}}
	\end{center}
	\caption{the periodic local configurations constructed from replaceable local configurations}
	\label{Fig5-3}
\end{figure}

\begin{theorem}
	\label{t2}
	\textbf{The non-existence of periodic local configurations is equivalent to global injectivity.}
\end{theorem}

This theorem is so concise that it can finish the determination for global injectivity perfectly. In the next subsection, the whole injectivity tree algorithm is introduced in detail.

\subsection{Injectivity tree algorithm and surjectivity with the periodic boundary}
First, in Amoroso's paper, a tree is used to decide the surjectivity of a CA. In essence, the tree uses paths as local configurations and tuples within nodes as a predecessor. After an empty node is found, or the tree cannot extend further, the construction of the tree is completed.

According to Theorem \ref{t2}, we need the information of the two sides of local configurations in the process of constructing the new tree for injectivity. Tuples in nodes record the right side of local configurations' predecessor. However, we cannot obtain the information of the left side of its predecessor in nodes. Therefore, the length of the tuple in the injectivity tree should be $2m-2$ bits, where the right $m-1$ bits of the tuple are constructed in a similar way to the surjectivity tree. If two tuples $\alpha=a_1a_2 \cdots a_{m-1}b_1b_2 \cdots b_{m-1}$ and $\beta=c_1c_2 \cdots c_{m-1}d_1d_2 \cdots d_{m-1}$ have $b_2b_3 \cdots b_{m-1}= d_1d_2 \cdots d_{m-2}$ and the node where $\alpha$ is located is the parent of the node where $\beta$ is located, then $\beta$ has the same left $m-1$ bit as $\alpha$ ($a_1a_2 \cdots a_{m-1}=c_1c_2 \cdots c_{m-1}$). This algorithm can also determine the surjectivity with the periodic boundary, The difference is that termination for injectivity is the periodic local configurations, and the termination of surjectivity with the periodic boundary is the Garden-of-Eden.

\begin{description}
	\item[Step 1] For a CA with neighborhood size $m$, get the local map table of the rule.
	\item[Step 2] Duplicate all $(m-1)$-tuples into $(2m-2)$-tuples. Add all these $(2m-2)$-tuples to the root node.
	\item[Step 3] For each node $M$ in layer $i$ $(i \geq 0)$, construct its children $M_0, M_1,\cdots, M_{p-1}$. For each tuple $a_1a_2 \cdots a_{2m-2}$ in $M$ and all $0 \leq j<p$, if $f(a_ma_{m+1} \cdots a_{2m-2}j)=t$, then $a_1a_2 \cdots a_{m-1}a_{m+1}a_{m+2} \cdots a_{2m-2}j$ is added to $M_t$. After each node is constructed, make a series of decisions. If the node is identical to the previous node, its children will not be constructed. If there is no tuple inside the node, then the CA has a Garden-of-Eden which decides the CA is neither injective nor surjective, and the algorithm is terminated immediately. If the node is not the root and there are two or more tuples inside, the two tuples $a_1a_2 \cdots a_{2m-2}$ and $b_1b_2 \cdots b_{2m-2}$ in it are both periodic. The CA is not injective, and the algorithm is terminated immediately.
	The following conditions can retrieve periodic local configurations.
	\begin{itemize}
		\item $a_1a_2 \cdots a_{m-1}= a_ma_{m+1} \cdots a_{2m-2}$
		\item $b_1b_2 \cdots b_{m-1}= b_mb_{m+1} \cdots b_{2m-2}$
	\end{itemize}
	
	\item[Step 4] If the three steps above are completed and the CA is not decided as non-injective, then the CA is injective.
\end{description}

Algorithm \ref{a3} shows its complete calculation process formally. This algorithm is very similar to Amoroso's algorithm for surjectivity.

\begin{algorithm}[h]
	\small
	\SetAlgoLined
	\label{a3}
	\KwData{local rule $f$}
	\KwResult{whether the CA is injective}
	let $L$ be an empty queue\;
	let $S_{root} \leftarrow \{(a_1...a_{m-1},a_1...a_{m-1}) | a_1,...,a_{m-1} \in \{0,1,...,p-1\}\}$ be the root set\;
	push $S_{root}$ to the back of $L$\;
	\While{$L$ is not empty}{
		remove the headset in $L$ and mark it as $S_{current}$\;
		\If{$S_{current}$ is empty}{
			return the CA is not injective\;
		}
		\If{$S_{current}$ includes more than one tuple $(a_1...a_{m-1},a_{m}...a_{2m-2})$ such that $a_1...a_{m-1}=a_{m}...a_{2m-2}$}{
			return the CA is not injective\;
		}
		\If{$S_{current}$ is a new set}{
			let $S_0$, $S_1$,...,$S_{p-1}$ be empty sets\;
			\For{each $(a_1a_2...a_{m-1},a_ma_{m+1}...a_{2m-2}) \in S_{current}$, each $d \in \{0,1,...,p-1\}$ and each $b \in \{0,1,...,p-1\}$}{
				\If{$f(a_ma_{m+1}...a_{2m-2}d) = b$} {
					add $(a_1...a_{m-1},a_{m+1}...a_{2m-2}d)$ into $S_b$\;
				}
			}
			push $S_0$, $S_1$,,...,$S_{p-1}$ to the back of $L$\;
		}
	}
	return the CA is injective\;
	\caption{decision algorithm for injectivity}
\end{algorithm}

\subsection{Example demonstration}

To show the injectivity tree algorithm in more detail, this subsection gives a construction example of the algorithm for rule 01100110. The three steps correspond to the first three steps above.
\begin{description}
	\item[Step 1] First, construct the rule table as shown in Figure \ref{Fig5-4}.
	\begin{figure}[h]
		\begin{center}
			\scalebox{0.04}{\includegraphics{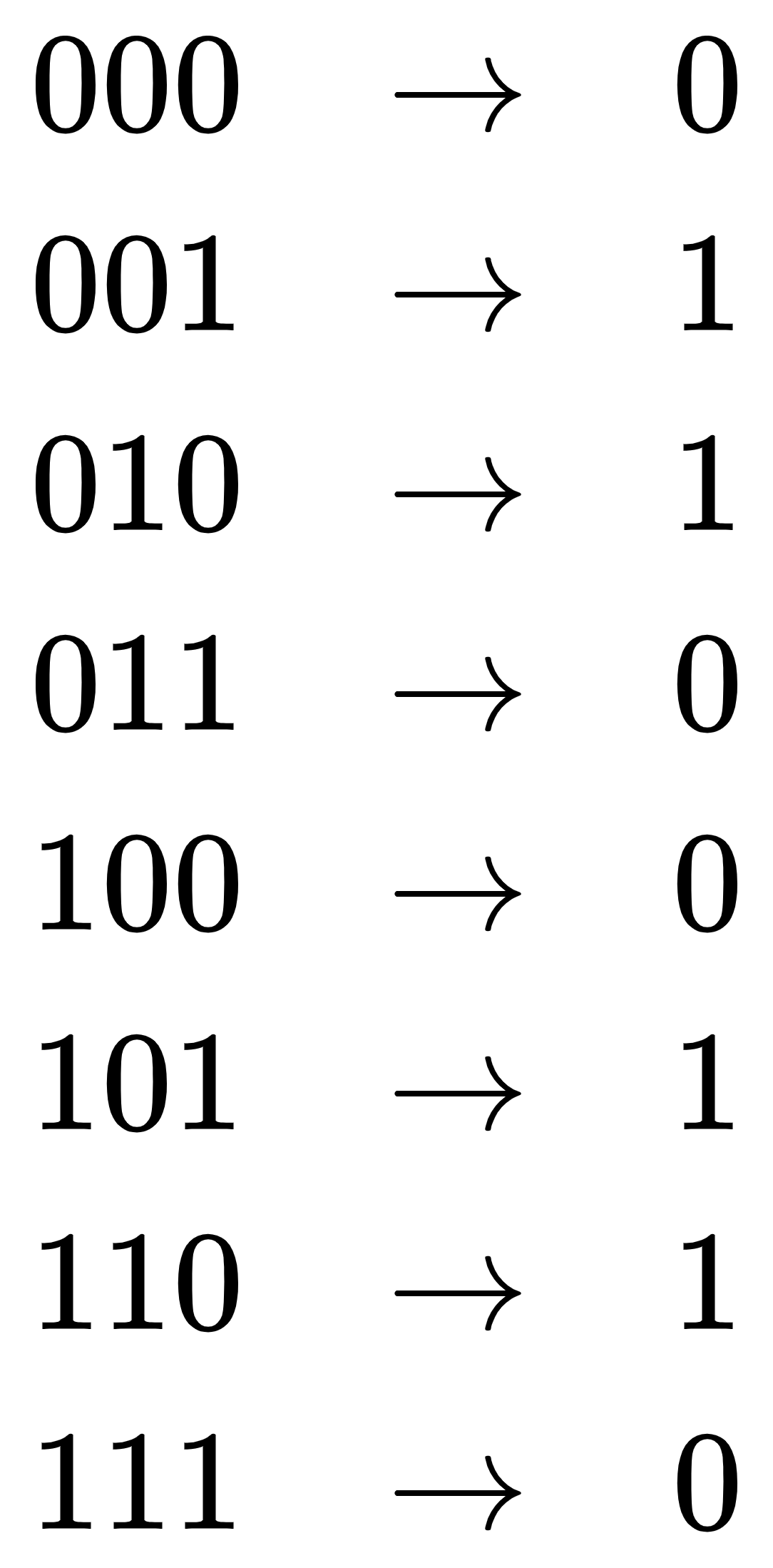}}
		\end{center}
		\caption{the rule table of CA with Wolfram number 01100110}
		\label{Fig5-4}
	\end{figure}
	\item[Step 2] Duplicate all 2-tuples into 4-tuples and add them to the root, as shown in Figure \ref{Fig5-5}. The first and last two bits of each tuple are the same.
	\begin{figure}[h]
		\begin{center}
			\scalebox{0.1}{\includegraphics{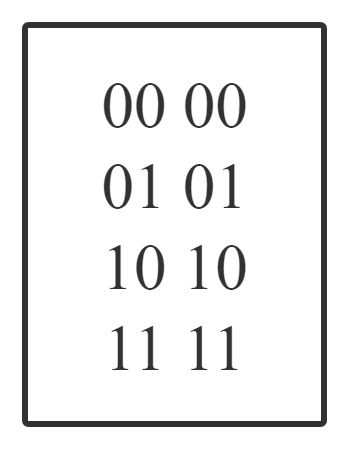}}
		\end{center}
		\caption{the construction of root}
		\label{Fig5-5}
	\end{figure}
	\item[Step 3] Constantly construct the children of each node with the given method. The whole injectivity tree is shown in Figure \ref{Fig5-6}. The algorithm is terminated rapidly because the tuples 00 00 and 11 11 are periodic. If we need to determine if this CA is surjective with the periodic boundary, we need to continue.
		\begin{figure}[h]
		\begin{center}
			\scalebox{0.1}{\includegraphics{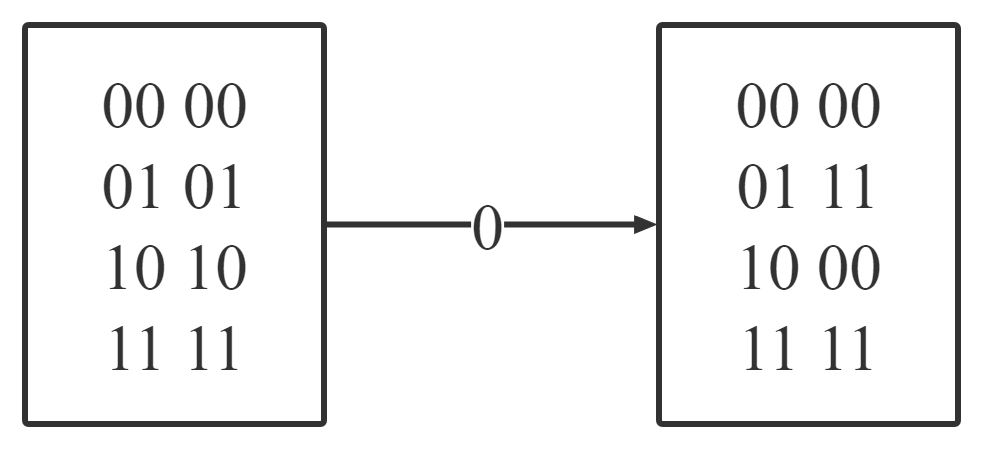}}
		\end{center}
		\caption{the injectivity tree}
		\label{Fig5-6}
	\end{figure}

\end{description}

\begin{figure}[h]
	\begin{center}
		\scalebox{0.1}{\includegraphics{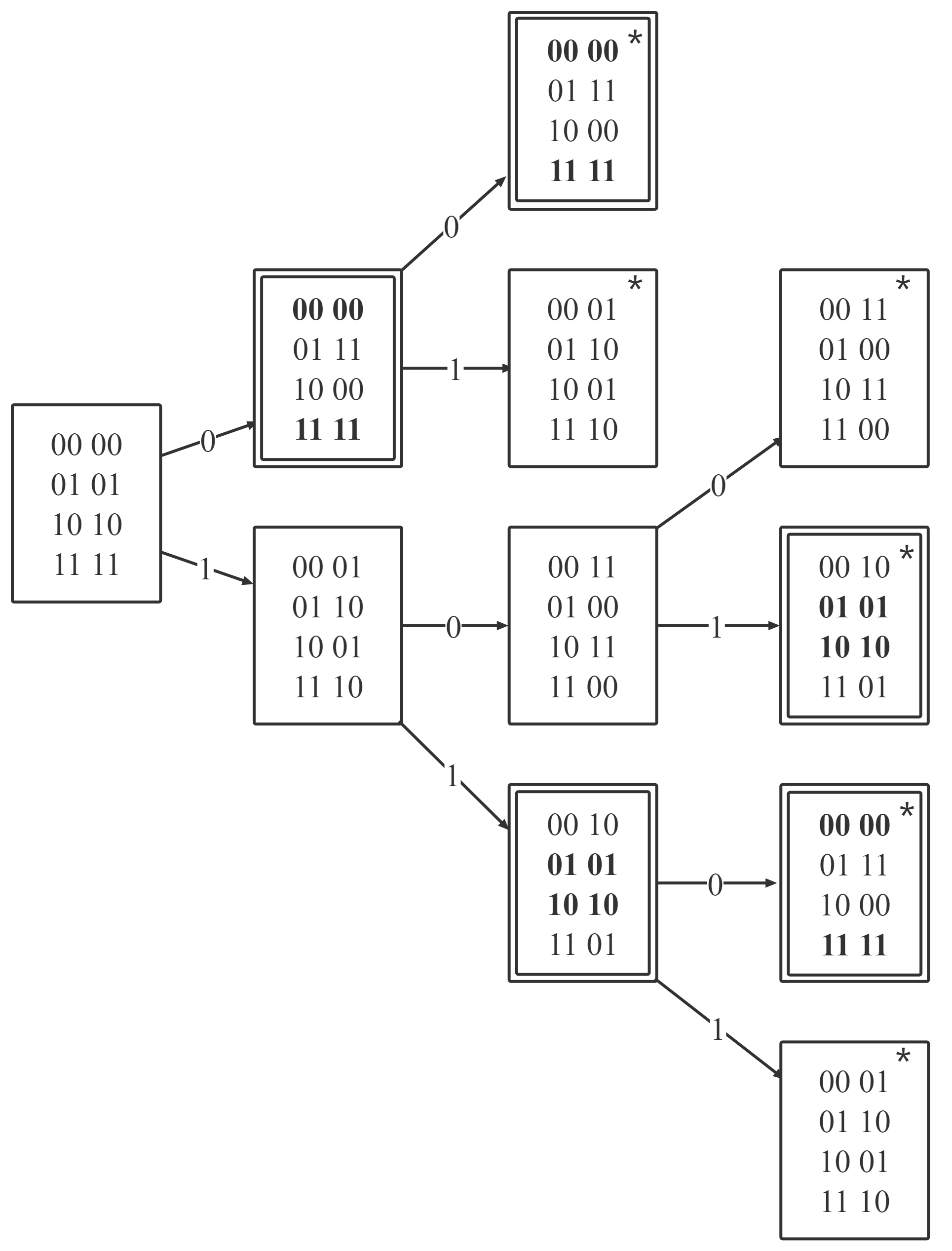}}
	\end{center}
	\caption{an injectivity tree without termination}
	\label{Fig5-7}
\end{figure}

To completely show this algorithm and its ability for surjectivity, we give an injectivity tree without non-injectivity termination, as shown in Figure \ref{Fig5-7}. The node with the symbol ``*" indicates that the node is the same as the one that appeared before. A node with double borders indicates that the node has at least one termination condition. For injectivity, we can find that although the overall structure of the tree may be enormous, the algorithm often terminates very quickly. On the other hand, this CA is a surjective, so it is larger than expected when determining the surjectivity. Hence, the complexity of the tree is much lower than expected. Figure \ref{Fig5-8} shows the tree of rule 110000110011100 with neighborhood size 4, and the results for surjectivity are shown in \ref{table1}.
\begin{figure}[h]
	\begin{center}
		\scalebox{0.08}{\includegraphics{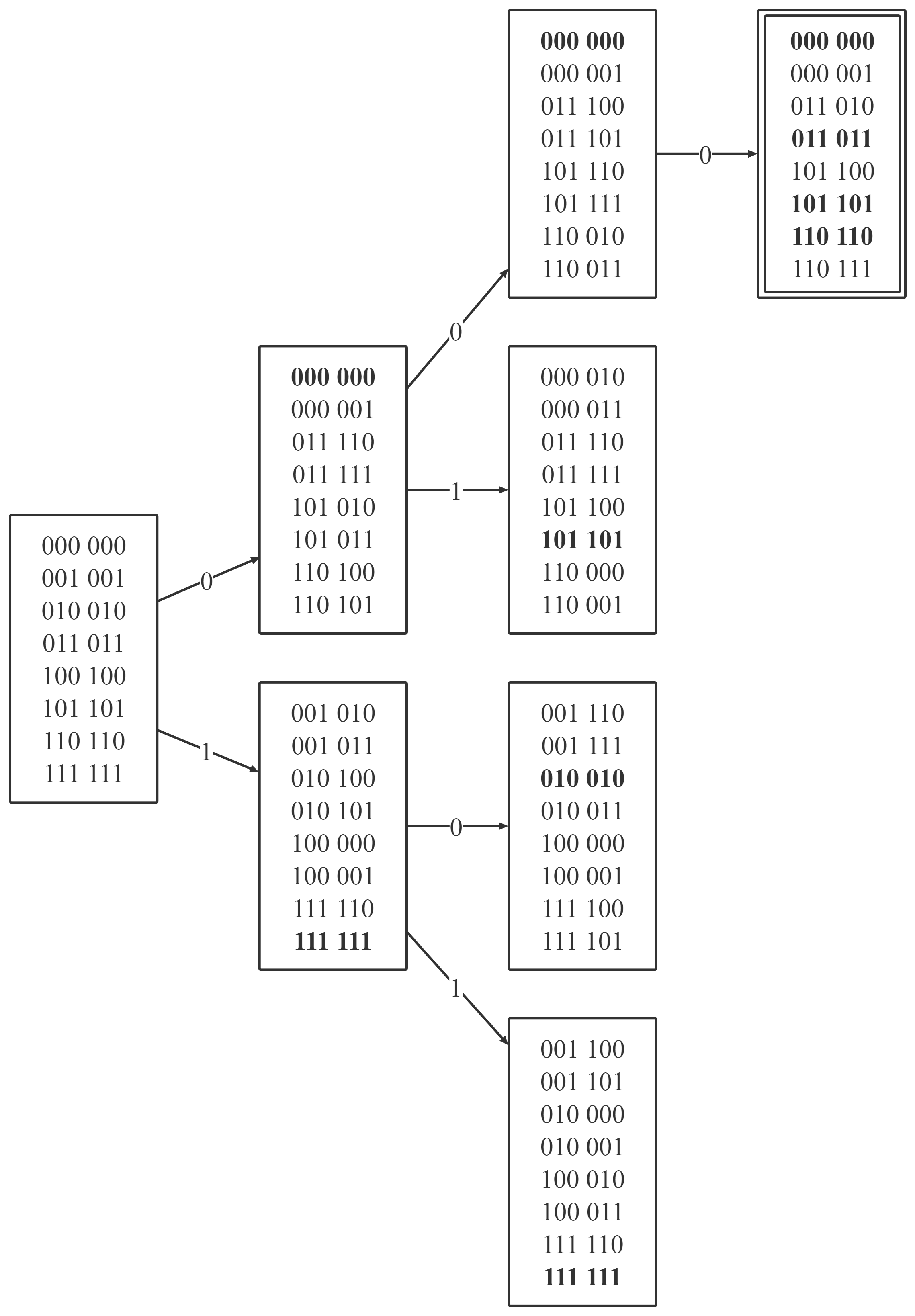}}
	\end{center}
	\caption{the injectivity tree of CA with Wolfram number 1100001100111100}
	\label{Fig5-8}
\end{figure}

\subsection{Theoretical complexity analysis}
Firstly, suppose the neighborhood size is $m$ and let $S=\{0, 1\}$. In the case of CA, the probability that each bit is $0$ and $1$ is equal. Therefore, if only $n$ consecutive cells are considered, the probability of $2^n$ local configurations is equal.

Secondly, this algorithm's complexity is the number of tuples in it. Each tuple can be represented by an integer which is usually regarded as the smallest unit in the algorithm. We only need to calculate the relationship between the number of tuples in the algorithm and the neighborhood size $m$.

We can divide all nodes into two categories: one with more than one tuple inside and the other with only one inside.

For the former category, assume the number of tuples in a node is $k$ $(k \geq 2)$. The length of tuples in the nodes is $2m-2$. Then each node will compare $k*(k-1)/2$ times to search periodic local configurations, and each comparison is independent of the other. On average, each tuple contributes $(k-1)/2$ comparisons. The probability of termination after each comparison is $P=1/2^{2m-2}$. The number of comparisons $t$ follows the geometric distribution. Mathematical expectation of $t$ is $E(t)= 2^{2m-2}$, so the program will be terminated after $2^{2m-2}$ comparisons on average. It has been calculated that each tuple will contribute $(k-1)/2$ comparisons on average, and the program will stop after $2^{2m-1}/(k-1)$ tuples on average. Because $k \geq 2$, $2^{2m-1}/(k-1)\leq 2^{2m-1}$. Therefore, the average number of tuples when the program terminates is less than $2^{2m-1}$.

For the latter category, there are at most $2^{2m-2}$ different single-tuple nodes on average. When they first appear, these nodes have $2^{2m-1}$ children. If they appear repeatedly, they do not have children. The nodes calculated in the previous paragraph also have $2^{2m-1}$ children. Now assume the number of single-tuple nodes reaches the maximum. The average number of single-tuple nodes shall be at most $2^{2m-2}+2^{2m-1}+2^{2m-1}*2=5*2^{2m-2}$. However, if we calculate it by probability, there is a 0.5 probability that each single-tuple node will have a Garden-of-Eden child, so the number of such nodes hardly needs to be considered during actual operation. Since we assume the worst-case scenario of $k=2$, the expected cost of time and space of this algorithm should be less than $2^{2m-1}$.It conforms to the quadratic time complexity proposed by Sutner (Sutner used $2^m$ as the variable). We continue to optimize the complexity of the algorithm based on this and give an exact and excellent expectation of its time and space consumption. The actual consumption is a half of the complexity.

\subsection{Actual running efficiency}
In many cases, the algorithm's complexity can not entirely represent the running efficiency of the algorithm. It is also closely related to the design ideas, data structures and other factors. Therefore, the actual operation is necessary for deciding whether an algorithm is excellent.

This paper compares the actual running time of the injectivity tree algorithm and Amoroso's injectivity algorithm from many aspects. For CA with a neighborhood size larger than 5, due to the small number of CA, it is possible to run all the CA (filter out the unbalanced CA in advance). For CA with a neighborhood size greater than 5, only some balanced CA are randomly selected for the comparative test because of the large total number of CA. The program runs in Java (JDK 1.8), and the running environment includes Windows 10, Eclipse, Maven, and JVM 1.8. The max heap size of the server is 512MB.

The running results of random balanced CA with neighborhood sizes 3 to 11 are shown in Table \ref{table2}. When the neighborhood size $(m)$ is small (3 and 4), there is little difference between the two algorithms. However, when the neighborhood size exceeds 6, there is a big gap between the two algorithms' operating efficiency. The larger the neighborhood size is, the more this gap will increase. Up to the neighborhood size of 11, the running time ratio of the two algorithms has reached impressive 21.481. Apart from time, Amoroso's injectivity algorithm is also out of memory with a neighborhood size of 12, while the injectivity tree algorithm can run successfully when the neighborhood size is 16.

\begin{table}[h]
	\center
	\caption{the comparison of two decision algorithms for injectivity}
	\label{table2}
	\begin{tabular}{cp{4em}p{8em}p{8em}c}
		\hline
		m & number & Average time in Amaroso's algorithm (ms) & Average  time in the injectivity tree algorithm (ms) & time ratio \\
		\hline
		3 & 70 & 0.0714 & 0.0571 & 1.250 \\

		4 & 12870 & 0.0153 & 0.0110 & 1.387 \\

		5 & 1000000 & 0.0384 & 0.0175 & 2.201 \\

		6 & 100000 & 0.157 & 0.0577 & 2.724 \\

		7 & 100000 & 0.621 & 0.149 & 4.157 \\

		8 & 10000 & 2.914 & 0.544 & 5.353 \\

		9 & 10000 & 15.559 & 1.946 & 7.994 \\

		10 & 1000 & 80.352 & 7.739 & 10.383 \\

		11 & 1000 & 623.436 & 29.023 & 21.481 \\

		12 & 1000 & out of memory & 110.900 & - \\

		13 & 1000 & out of memory & 424.911 & - \\

		16 & 10 & out of memory & 19758.5 & - \\
		\hline
	\end{tabular}
	
\end{table}

Amoroso's algorithm consumes a high proportion of time in step 4 (iterative deletion of crossed-out boxes). If the deletion is performed only through traversal, the complexity will be close to $O(2^{4m})$. The complexity needs to be reduced to quadratic time using a hashmap or other data structures. However, in step 4, we still need to traverse all sequent sets and their sequent sets, which increase exponentially with neighborhood size. It is unnecessary because some sequent sets cannot exist in infinite configurations. Sutner's algorithm is the same. They all need to spend time constructing a data structure (graph or table) and then spend more time disassembling and deleting the data structure element by element. In the injectivity tree algorithm, only the tuples that may appear in the infinite configurations are considered without deleting any nodes in the tree. If we can complete the determination of injectivity when constructing the data structure (tree), we will save most of the cost of deletion operations. That is another reason that the injectivity tree algorithm is the best.

Here, we can use a more common problem for analogy, the minimum spanning tree. Both the Prim algorithm and the breaking method can get correct results for this problem. Although there is no apparent difference in complexity between the breaking method and the Prim algorithm, the breaking method is rarely selected in actual operation. Because it is often more expensive to retrieve and delete a data structure than to build a data structure. What's more, the efficiency of the breaking method is not satisfactory without building a data structure. It is not difficult to understand why the complexity of the two algorithms is the same, but the actual operation efficiency is greatly different.

Secondly, there are multiple tuples in a node, which is more efficient than storing two tuples per storage unit in the algorithms of Amoroso and Sutner. We can retrieve all tuples in a node in linear time $k$ and perform a simple count, while the retrieval cost in the previous algorithms is $k*(k-1)/2$.

Next, this decision with trees is a linear process. After finding the termination, the algorithm can be terminated immediately instead of deleting all crossed-out boxes iteratively. For example, in Figure \ref{Fig5-5}, the tree ends rapidly because it does not take other sequent sets into consideration. The above explains why the injectivity tree algorithm has better time performance. In addition, Amoroso's algorithm needs to use a hashmap to reduce its complexity. The hashmap will cost a lot of space, leading to out-of-memory in the case of small neighborhood sizes. In conclusion, the space and time efficiency of the injectivity tree algorithm in the actual running is the best, it has reached the \textbf{state-of-art}.

\section{Injectivity with different boundaries \label{s5}}
Another advantage of the injectivity tree algorithm is that it can be easily extended to the boundary conditions. Comparing it with bounded surjectivity, we have some new conclusions. Since there is little difference between the reflective boundary and the fixed boundary, and their calculation process is almost the same, the reflective boundary will not be discussed in this section.

\begin{theorem}
	\label{t3}
	\textbf{The injectivity and surjectivity are equivalent with the same boundary.}
\end{theorem}

\begin{proof}
	Suppose that the state set $S=\{0,1,\ldots, p-1\}$ and the number of cells is $n$. Because of the limited boundary, for each $n$, the elements of input and output sets are equal. So, injectivity and surjectivity are equivalent, or in other words, it is bijective and reversible.
\end{proof}

\subsection{Injectivity with fixed boundary}
The injectivity with fixed boundary is similar to the Subsection \ref{s32}. The difference is that the surjectivity tree is replaced by the injectivity tree. For the fixed boundary, the initial set and terminal set should be configured and because of the double length of the tuples, the elements in these two sets should also be doubled. Because of the definition of the fixed boundary, the local configurations within the boundary are always replaceable. So what we should search for becomes the replaceable configurations.

\begin{lemma}
	\label{l5}
	The following descriptions are equivalent:
	\begin{itemize}
		\item The CA is surjective with the fixed boundary.
		\item The CA is injective with the fixed boundary.
		\item There exists a Garden-of-Eden in the surjectivity tree.
		\item There exists replaceable local configurations in the injectivity tree.
	\end{itemize}
\end{lemma}

\subsection{Injectivity with the periodic boundary}

For injectivity with the periodic boundary, we need not configure the initial set and terminal set. All we need to do is to search periodic local configurations. We can be surprised to find that the global injectivity is equivalent to the injectivity with the periodic boundary, for they have the same decision algorithm! With this theorem, we can equalize injectivity and reversibility.

\begin{theorem}
	\label{t4}
	\textbf{The reversibility with the periodic boundary and the global injectivity is equivalent.}
\end{theorem}

\begin{proof}
	If a CA is global injective, there will not be periodic local configurations according to theorem \ref{t2}. So it is reversible under the periodic boundary condition. Similarly, if a CA is not global injective, there will be periodic local configurations. Then there is a certain configuration is not injective  under the periodic boundary condition. So it is not reversible under the periodic boundary condition.
\end{proof}

We also searched all CA with neighborhood sizes less than 6 and found that they did comply with this phenomenon.
Now we can use this decision algorithm for injectivity to determine the reversibility of freezing CA \cite{2015_Goles, 2020_Goles,2020_Goles2,2021_Goles} and number-conserving CA \cite{2017_Wolnik,2020_Wolnik,2022_Wolnik} and so on which are under periodic boundary.

\begin{theorem}
	\label{t5}
	\textbf{The bijectivity of an LCA with a certain boundary is equivalent to that all its transition matrices with this boundary are reversible.}
\end{theorem}

%

With the help of Theorem \ref{t5}, we do not need to discuss the bijective LCA's reversibility with boundaries anymore. All of them are reversible.

\section{Conclusion and future work \label{s6}}
\subsection{Conclusion}
In this paper, we greatly simplify Amoroso's surjectivity algorithm and extend it to the case with different boundaries. And a new decision algorithm for injectivity is proposed which greatly saves the time and space required for injectivity determination, which is caused by the excellent design and structure of this algorithm. In practical applications, the more complex the CA is, the better the simulation, encryption and other applications are. This algorithm can help us expand the calculation of one-dimensional CA to a more complex scale, thus achieving better results in practical applications. We also analyzed the boundary problem from a nonlinear point of view, solved the connection between the surjectivity and injectivity under different boundary conditions, and obtained the important equivalence relationship from the condition that their decision algorithms are the same. We also extend our work from linearity to nonlinearity and solve the reversibility problem of linear CA under a large class of boundary conditions through nonlinear decision algorithms. In addition, the various definitions, lemmas, theorems, algorithms and even logic and ideas proposed in the algorithm can help us further explore the characteristics of CA, provide a theoretical basis for future applications, and provide inspiration for future theoretical research.

\subsection{Future work}
During our research, three questions emerged. The first one concerns that in Table \ref{table1}, we can find the setting of the boundary does not affect the surjectivity of the CA. At least there is no exception for CA below neighborhood size 5. So for any fixed boundary CA, whether the setting of the boundary will never affect the surjectivity (reversibility) of the CA?
\begin{question}
	For any fixed boundary CA, whether the setting of the boundary will never affect the surjectivity (reversibility) of the CA?
\end{question}

The second question is connected with the reflective boundary in Table \ref{table1}, we can find all only two CA are surjective (the CA with identity rule and its opposite). So besides the identity CA, are there any other surjective CA with the reflective boundary?

\begin{question}
	Are there any other surjective CA with reflective boundary except for the CA with identity rule and its opposite?
\end{question}

The third question is based on Wang's work \cite{2015_ChaoWang}. We can attain the period of LCA with fixed boundary by constructing its DFA, so can we construct the DFA for periodic boundary and reflective boundary?

\begin{question}
	Can we construct the DFA of CA compatible with the periodic boundary and the reflective boundary and find its period?
\end{question}

\bibliography{ijuc}

\begin{thebibliography}{10}

\bibitem{2012_Akin}
Hasan Akin, Ferhat Sah, and Irfan Siap.
\newblock (04 2012).
\newblock On 1d reversible cellular automata with reflective boundary over the
  prime field of order p.
\newblock {\em International Journal of Modern Physics C}, 23:1250004.

\bibitem{2014_Akin}
Hasan Akin, Irfan Siap, and Su~Uğuz.
\newblock (02 2014).
\newblock One-dimensional cellular automata with reflective boundary conditions
  and radius three.
\newblock {\em Acta Physica Polonica Series a}, 125:405--407.

\bibitem{2011_Akin}
Hasan Akin, Su~Uğuz, and Irfan Siap.
\newblock (09 2011).
\newblock Characterization of 2d cellular automata with moore neighborhood over
  ternary fields.
\newblock {\em AIP Conference Proceedings}, 1389:2008--2011.

\bibitem{1975_Amoroso}
Serafino Amoroso, Gerald Cooper, and Yale Patt.
\newblock (1975).
\newblock Some clarifications of the concept of a garden-of-eden configuration.
\newblock {\em Journal of Computer and System Sciences}, 10(1):77--82.

\bibitem{1972_Amoroso}
Serafino Amoroso and Yale~N. Patt.
\newblock (1972).
\newblock Decision procedures for surjectivity and injectivity of parallel maps
  for tessellation structures.
\newblock {\em Journal of Computer and System Sciences}, 6(5):448--464.

\bibitem{2007_Bingham}
Jesse Bingham and Brad Bingham.
\newblock (2007).
\newblock Hybrid one-dimensional reversible cellular automata are regular.
\newblock {\em Discrete Applied Mathematics}, 155(18):2555--2566.

\bibitem{1979_Bruckner}
L.~K. Bruckner.
\newblock (1979).
\newblock On the garden-of-eden problem for one-dimensional cellular automata.
\newblock {\em Acta Cybern.}, 4(3):259--262.

\bibitem{1987_Culik}
Karel Cul{\'i}k.
\newblock (1987).
\newblock On invertible cellular automata.
\newblock {\em Complex Syst.}, 1:1035--1044.

\bibitem{2006_Rey}
A.~Martín {del Rey} and G.~{Rodrı´guez Sánchez}.
\newblock (2006).
\newblock On the reversibility of 150 wolfram cellular automata.
\newblock {\em International Journal of Modern Physics C}, 17:975--983.

\bibitem{2011_Rey}
A.~Martín {del Rey} and G.~{Rodrı´guez Sánchez}.
\newblock (2011).
\newblock Reversibility of linear cellular automata.
\newblock {\em Applied Mathematics and Computation}, 217(21):8360--8366.

\bibitem{2015_Rey}
A.~Martín {del Rey} and Gerardo Sánchez.
\newblock (02 2015).
\newblock Reversible elementary cellular automaton with rule number 150 and
  periodic boundary conditions over f{double struck}p.
\newblock {\em International Journal of Modern Physics C}, 26:150210231125000.

\bibitem{2022_ChaoWang}
Xinyu Du, Chao Wang, Tianze Wang, and Zeyu Gao.
\newblock (2022).
\newblock Efficient methods with polynomial complexity to determine the
  reversibility of general 1d linear cellular automata over zp.
\newblock {\em Information Sciences}, 594:163--176.

\bibitem{2007_Rey}
L.~Hernández Encinas and A.~Martín {del Rey}.
\newblock (2007).
\newblock Inverse rules of eca with rule number 150.
\newblock {\em Applied Mathematics and Computation}, 189(2):1782--1786.

\bibitem{1970_Convey}
Martin Gardner.
\newblock (1970).
\newblock Mathematical games.
\newblock {\em Scientific american}, 222(6):132--140.

\bibitem{2020_Goles}
Eric Goles, Diego Maldonado, Pedro Montealegre, and Nicolas Ollinger.
\newblock (2020).
\newblock On the complexity of the stability problem of binary freezing
  totalistic cellular automata.
\newblock {\em Information and Computation}, 274:104535.
\newblock AUTOMATA 2017.

\bibitem{2021_Goles}
Eric Goles, Diego Maldonado, Pedro Montealegre, and Martín Ríos-Wilson.
\newblock (2021).
\newblock On the complexity of asynchronous freezing cellular automata.
\newblock {\em Information and Computation}, 281:104764.

\bibitem{2020_Goles2}
Eric Goles and Pedro Montealegre.
\newblock (2020).
\newblock The complexity of the asynchronous prediction of the majority
  automata.
\newblock {\em Information and Computation}, 274:104537.
\newblock AUTOMATA 2017.

\bibitem{2015_Goles}
Eric Goles, Nicolas Ollinger, and Guillaume Theyssier.
\newblock (2015).
\newblock Introducing freezing cellular automata.
\newblock In {\em Cellular Automata and Discrete Complex Systems, 21st
  International Workshop (AUTOMATA 2015)}, volume~24, pages 65--73.

\bibitem{1989_Head}
Tom Head.
\newblock (1989).
\newblock One-dimensional cellular automata: injectivity from unambiguity.
\newblock {\em Complex systems}, 3:343--348.

\bibitem{1983_Ito}
Masanobu Itô, Nobuyasu Ôsato, and Masakazu Nasu.
\newblock (1983).
\newblock Linear cellular automata over zm.
\newblock {\em Journal of Computer and System Sciences}, 27(1):125--140.

\bibitem{1994_Kari}
Jarkko Kari.
\newblock (1994).
\newblock Reversibility and surjectivity problems of cellular automata.
\newblock {\em Journal of Computer and System Sciences}, 48(1):149--182.

\bibitem{2005_Kari}
Jarkko Kari.
\newblock (2005).
\newblock Theory of cellular automata: A survey.
\newblock {\em Theoretical Computer Science}, 334(1):3--33.

\bibitem{1962_Moore}
Edward~F. Moore.
\newblock (1962).
\newblock Machine models of self-reproduction.
\newblock volume~14, pages 17--33.

\bibitem{1963_Myhill}
John~R. Myhill.
\newblock (1963).
\newblock The converse of moore’s garden-of-eden theorem.
\newblock volume~14, pages 685--686.

\bibitem{2004_Nobe}
Atsushi Nobe and Fumitaka Yura.
\newblock (may 2004).
\newblock On reversibility of cellular automata with periodic boundary
  conditions.
\newblock {\em Journal of Physics A: Mathematical and General}, 37(22):5789.

\bibitem{2019_Rey}
Ángel Rey, Roberto Casado-Vara, and Daniel Hernández~Serrano.
\newblock (09 2019).
\newblock Reversibility of symmetric linear cellular automata with radius r =
  3.
\newblock {\em Mathematics}, 7:816.

\bibitem{2008_Sahoo}
Sudhakar Sahoo, Sanjaya Sahoo, Birendra Nayak, and Pabitra Choudhury.
\newblock (09 2008).
\newblock Encompression using two-dimensional cellular automata rules.

\bibitem{1991_Sutner}
Klaus Sutner.
\newblock (1991).
\newblock De bruijn graphs and linear cellular automata.
\newblock {\em Complex Syst.}, 5:19--30.

\bibitem{1951_Neumann}
John von Neumann and Arthur~W. Burks.
\newblock (1966).
\newblock Theory of self reproducing automata.

\bibitem{1983_Wolfram}
Stephen Wolfram.
\newblock (1983).
\newblock Statistical mechanics of cellular automata.
\newblock {\em Reviews of modern physics}, 55(3):601.

\bibitem{1984__Wolfram}
Stephen Wolfram.
\newblock (1984).
\newblock Computation theory of cellular automata.
\newblock {\em Communications in mathematical physics}, 96(1):15--57.

\bibitem{1984_Wolfram}
Stephen Wolfram.
\newblock (1984).
\newblock Universality and complexity in cellular automata.
\newblock {\em Physica D: Nonlinear Phenomena}, 10(1-2):1--35.

\bibitem{1985_Wolfram}
Stephen Wolfram.
\newblock (jan 1985).
\newblock Twenty problems in the theory of cellular automata.
\newblock {\em Physica Scripta}, 1985(T9):170.

\bibitem{1986_Wolfram}
Stephen Wolfram.
\newblock (1986).
\newblock Theory and applications of cellular automata.
\newblock {\em World Scientific}.

\bibitem{2002_Wolfram}
Stephen Wolfram and M~Gad-el Hak.
\newblock (2003).
\newblock A new kind of science.
\newblock {\em Appl. Mech. Rev.}, 56(2):B18--B19.

\bibitem{2020_Wolnik2}
Barbara Wolnik and Bernard {De Baets}.
\newblock (2020).
\newblock Ternary reversible number-conserving cellular automata are trivial.
\newblock {\em Information Sciences}, 513:180--189.

\bibitem{2017_Wolnik}
Barbara Wolnik, Adam Dzedzej, Jan Baetens, and Bernard De~Baets.
\newblock (05 2017).
\newblock Number-conserving cellular automata with a von neumann neighborhood
  of range one.
\newblock {\em Journal of Physics A: Mathematical and Theoretical}, 50.

\bibitem{2022_Wolnik}
Barbara Wolnik, Maciej Dziemiańczuk, Adam Dzedzej, and Bernard {De Baets}.
\newblock (2022).
\newblock Reversibility of number-conserving 1d cellular automata: Unlocking
  insights into the dynamics for larger state sets.
\newblock {\em Physica D: Nonlinear Phenomena}, 429:133075.

\bibitem{2020_Wolnik}
Barbara Wolnik, Anna Nenca, Jan~M. Baetens, and Bernard {De Baets}.
\newblock (2020).
\newblock A split-and-perturb decomposition of number-conserving
  cellular automata.
\newblock {\em Physica D: Nonlinear Phenomena}, 413:132645.

\bibitem{2015_ChaoWang}
Bin Yang, Chao Wang, and Aiyun Xiang.
\newblock (2015).
\newblock Reversibility of general 1d linear cellular automata over the binary
  field z2 under null boundary conditions.
\newblock {\em Information Sciences}, 324:23--31.

\end{thebibliography}

\appendix 

\end{document}